\pdfoutput=1

\documentclass{article}
\usepackage{ijcai25}

\usepackage{times}
\usepackage{soul}
\usepackage{url}
\usepackage[hidelinks]{hyperref}
\usepackage[utf8]{inputenc}
\usepackage[small]{caption}
\usepackage{graphicx}
\usepackage{amsmath}
\usepackage{amsthm}
\usepackage{booktabs}
\usepackage{algorithm}
\usepackage{algorithmic}

\usepackage{amssymb}
\usepackage{subcaption}
\usepackage{multirow}
\usepackage{comment}
\usepackage{cleveref}
\usepackage{float}
\usepackage{thmtools}
\usepackage{thm-restate}
\usepackage{newfloat}
\usepackage{listings}
\floatstyle{ruled}
\newfloat{listing}{tb}{lst}{}
\floatname{listing}{Listing}

\newcommand{\ips}{\bar{\tau}}
\newcommand{\lps}{\tau}

\newcommand{\R}{\mathbb{R}}

\newtheorem{theorem}{Theorem}

\title{Large-Scale Trade-Off Curve Computation for Incentive Allocation with Cardinality and Matroid Constraints}

\author{
Yu Cong$^1$, Chao Xu$^1$ and Yi Zhou$^1$
\affiliations
$^1$University of Electronic Science and Technology of China
\emails
\{yucong143, the.chao.xu\}@gmail.com, zhou.yi@uestc.edu.cn
}

\begin{document}

\maketitle
\begingroup
\renewcommand{\thefootnote}{}
\footnotetext{Supplementary materials are available at \url{https://github.com/congyu711/incentive-allocation-supplementaries}.}
\endgroup

\begin{abstract}
    We consider a large-scale incentive allocation problem where the entire trade-off curve between budget and profit has to be maintained approximately at all time. 
    The application originally comes from assigning coupons to users of the ride-sharing apps, where each user can have a limit on the number of coupons been assigned. We consider a more general form, where the coupons for each user forms a matroid, and the coupon assigned to each user must be an independent set. We show the entire trade-off curve can be maintained approximately in near real time.
  \end{abstract}

  \section{Introduction}

In the current age, we are dealing with increasingly large incentive allocation problems. One is given a fixed amount of budget to allocate to different incentives to maximize some objective. A prototypical example is assigning a single coupon to each rider in ridesharing apps, where each assignment uses up some marketing budget, and increase some metric such as rides or driving hours \cite{wang_shmoys_2019}. For example, an incentive allocation problem under cardinality constraints can be formalized as the following integer program.

\begin{align*}
    \max_x \quad \sum_{i}\sum_{j} v_{ij}x_{ij}\\
    \text{s.t.}   \quad \sum_{i}\sum_{j} c_{ij}x_{ij} &\leq B\\
    \sum_j x_{ij}&\leq k\phantom{\in \{0,1\}} \quad \forall i\\
    x_{ij}&\in \{0,1\}\phantom{\leq k} \quad \forall i, \forall j\\
\end{align*}

For each agent $i$, there is a candidate incentive set consisting of coupons, where the value and cost of the $j$th coupon is $v_{ij}$ and $c_{ij}$, respectively. The goal is to select at most $k$ coupons for each agent and to make sure that the total cost of the selected coupons does not exceed the budget $B$ while maximizing the total value of those coupons.

The problem is a variant of the knapsack problem, and computing the exact optimum is NP-hard. However, the fractional optimum is very close to the integer optimum, even if there are strong constraints in the allocation \cite{CAMERINI1984157}. Hence, in this article, we only consider finding the fractional optimum.

The allocation under a fixed budget is often insufficient for decision and analytic purposes. For example, the company might want to decide on the total budget for a campaign. A data scientist might need to know how much marketing spend is required to obtain an expected profit. These questions can all be answered if the entire trade-off curve of the budget vs profit can be computed.
There can be more complicated cases where the decider might not be a person but an algorithm. Consider the following case: The ridesharing company wants to allocate a budget to two campaigns, one of which is incentive allocation. We are optimizing $\max_x f(x)+g(B-x)$, where $f$ would map to the trade-off curve in the incentive allocation problem and $g$ is the value of the other campaign. $g$ can be complicated and in that case an algorithm optimizing the sum would evaluate $g$ at many different values of $x$. Implementing such an algorithm will be a lot faster and easier if we can compute the trade-off curve $f$ quickly.
Moreover, the curve is not static: In practice, the cost and value of a coupon are usually predicted by algorithms or models. An agent might take a certain action, say take a ride, and the model would change its predictions of the expected profit of each incentive associated with the agent. 

Hence, we investigate the \emph{dynamic} incentive allocation trade-off curve problem, where the \emph{entire} trade-off curve has to be maintained, while supporting updates(insertions and deletions) of agents' incentives. 
In practice, the number of agents is large (in the millions), the choices of incentives for each agent is relatively small (a few hundred), and no agent is critical to the objective. That is, removing any agent would not significantly impact the objective. Also, we assume each agent is independent: the incentives to one do not affect others. 

Generally, for each agent, there can be constraints on the allocation of incentives. We consider some examples in the ridesharing apps. A user can be assigned at most one of the incentives (\emph{multiple choice constraint}). A user can be assigned no more than $p$ incentives (\emph{cardinality constraint}). A user can have $2$ incentives for weekends, and $2$ incentives for weekdays, but only $3$ incentives in total (special case of \emph{matroid constraint}). The most general constraints are given as an arbitrary family of feasible subset of incentives. Our work would also consider how the problem changes under different constraints, but we mainly focus on cardinality and matroid constraints. 

Finally, we want the implementation to be easily transferable to queries in a modern OLAP databases.

\begin{table*}[!htb]
\centering
    \begin{tabular}{ccccc}
         Constraint Type & Result & Fixed budget & Trade-off curve & Dynamic\\
         \bottomrule
         \hline
         \multirow{3}{*}{Multiple Choice}& \cite{Dyer84,ZEMEL1984123}& $O(m)$ & -  & -\\
         &\cite{10.1109/ITSC55140.2022.9922143} & - & $O(m\log m)$  & No\\
         & \Cref{thm:cardinality} & - & $O(m\log m)$ & Yes\\
         \hline
         \multirow{4}{*}{Cardinality}& \cite{DavidPisinger} & $O(m\log VC)$  & -  & -\\
           & \cite{DavidPisinger} & $O(mp+nB)$  & -  & -\\
          & \cite{minimaxoptimization} & $O(m\log m)$  & -  & -\\
          & \Cref{thm:cardinality} & - & $O((k+m)\log m)$ & Yes\\
         \hline
         \multirow{3}{*}{Matroid}& \cite{CAMERINI1984157} & $O(m^2 + T \log m)$ & - & -\\
        & \cite{minimaxoptimization} & $O(T \log m)$ & - & -\\
        & \Cref{thm:matroid} & - & $O(Tk+k\log m)$ & Yes\\
        \bottomrule
    \end{tabular}
\caption{Comparison of algorithms for incentive allocation: $m$ is the total number of incentives, $M$ is the maximum number of incentives over each agent, $p$ is the max rank of the matroid constraint over each agent, or the limit in the cardinality constraint. $V$ and $C$ is the maximum value and cost of the incentives, respectively. $B$ is the budget. $k=O(mp^{1/3})$ bounds the number of breakpoints of the trade-off curve; in the cardinality bound we take $k=mp^{1/3}$. $T$ is the time complexity of matroid optimum base algorithm.}
\label{runtimetable}
\end{table*}

\paragraph{Previous Works.}
The (integral) incentive allocation problem for a fixed budget $B$ is a knapsack problem with side constraints. Our work is concerned with the linear programming relaxation, the fractional version, of the knapsack problem. 

When each agent is allocated exactly $1$ incentive, it is also called the \emph{continuous multiple choice knapsack problem} (CMCKP), and was widely studied. Sinha and Zoltners \cite{Zoltners} showed the optimum gap from the integral case is the value of a single incentive. Later, optimum linear time algorithm was discovered \cite{Dyer84,ZEMEL1984123}.
When each agent is required to be allocated exactly (or at most) $p$ incentives, namely having cardinality constraint on the incentives, it is equivalent to the \emph{continuous bounded multiple choice knapsack problem} (CBMCKP). CMCKP is the special case of CBMCKP when $p=1$.  Pisinger showed a reduction from CBMCKP to CMCKP, but running time depends on $B$ \shortcite{DavidPisinger}. In the same paper, Pisinger used the Dantzig-Wolfe decomposition to devise a faster polynomial time algorithm. However, the algorithm's running time depends on the size of the value and the cost, therefore it is not a strongly polynomial time algorithm. 

When the incentive for each agent must form an independent set (or a base) in a matroid, it is the (continuous) matroidal knapsack problem \cite{CAMERINI1984157}. The running time for finding an optimum is $O(m^2+ T\log m)$ time, where $m$ is the number of incentives and $T$ is the complexity of finding the optimum base for a given weighting of the elements in the matroid. After the technique of parametric search was introduced and improved \cite{Megiddo,Cole87}, the running time was improved to $O(T\log m)$ \cite{minimaxoptimization}. CBMCKP is a special case of the matroidal knapsack problem when the matroid is a $p$-uniform matroid. Although not explicitly stated, the matroid algorithm can be used for CBMCKP, and obtain an $O(m\log m)$ time algorithm because it takes $O(m)$ time to find the optimum base for a uniform matroid \cite{minimaxoptimization}. See \Cref{runtimetable} for a comparison of results. 

From another point of view, the incentive allocation problem can be considered a matroid optimization problem with an additional linear constraint. For general matroid this problem admits no fully polynomial-time approximation scheme \cite{doronarad_et_al:LIPIcs.ICALP.2024.56}, but admits an EPTAS \cite{DoronAradKulikShachnai2026}. For laminar matroids, an FPTAS is known \cite{DoronAradKulikShachnai2023}, with an improved running time in \cite{YangKovalyovZhang2025}. These approximation results concern integral solutions for a fixed budget.

For readers familiar with parametric or multi-objective optimization, it may also be helpful to view the trade-off curve as the Pareto curve between objectives. Under the multi-objective optimization framework, we are solving matroidal knapsack problem with an additional objective that minimize the total budget. Computing the trade-off curve can also be considered a sensitivity analysis problem, where the budget is the parameter whose sensitivity we are interested in. While these interpretations provide additional insight, our analysis is mainly conducted within the linear programs for the incentive allocation problem, as LPs better capture the properties of the problem and are easier to understand.

We are not aware of explicit computation of the entire trade-off curve except in the CMCKP case. A recent study in the transportation economics area \cite{10.1109/ITSC55140.2022.9922143} considered each agent must pick one of a few incentives, each having a different impact to social welfare. The regulator consults the entire trade-off curve for informed policymaking. The algorithm has a running time of $O(m\log m)$. The result is static, as it does not concern about updating the curve when individual incentive changes. 

\paragraph{Our contribution.}

We show that the entire incentive allocation trade-off curve is piecewise linear and concave. We construct a conceptually simple method to maintain the curve under different constraints, while allowing updates in logarithmic time with respect to number of fundamental changes of the trade-off curve.
In particular,
\begin{enumerate}
    \item In the multiple choice constraint case, the result matches the current fastest algorithm for static trade-off curve, but our implementation allows dynamic updates.
    \item In the cardinality constraint case, we show the \emph{entire trade-off curve} can be computed in $O(mp^{1/3}\log m)$ time.
    \item We also observe that our problem is related to the $k$-level problem in computational geometry and parametric matroid optimization. The connection shows a subquadratic bound to the number of breakpoints in the trade-off curve, when previously the bound is quadratic.
\end{enumerate}

Finally, we show part of the algorithm can be handled by modern OLAP database to avoid implementation complexity.

As a preview, we will prove the following theorems.

\begin{restatable}{theorem}{cardinality}
    \label{thm:cardinality}
    Consider an incentive allocation problem with a total of $m$ incentives. 
    If there is a cardinality $p\geq 1$ constraint on each agent, let $k=mp^{1/3}$. The trade-off curve has $O(k)$ breakpoints and can be computed in $O((k+m)\log m)=O(mp^{1/3}\log m)$ time.
\end{restatable}

\begin{restatable}{theorem}{matroid}
    \label{thm:matroid}
    Consider an incentive allocation problem with a total of $m$ incentives. 
    If there is a matroid constraint on each agent, each matroid has rank at most $p$, and $k$ upperbounds the number of breakpoints on the trade-off curve, then $k=O(mp^{1/3})$, and the trade-off curve can be computed in $O(Tk+k\log m)$ time, where $T$ is the time to compute an optimum weight base.
\end{restatable}


\begin{restatable}{theorem}{updatethm}
    \label{thm:update}
    If the slope-difference form of trade-off curve after an update differs from previous trade-off curve at $t$ positions, then the update takes $O(t\log k)$ time, where $k$ is the total number of breakpoints in the curve.
\end{restatable}

Assuming each agent is only available for a few hundred incentives, then each update of an agent, $t$ would be around the same number, which would make the running time near-real time.  

\section{Preliminaries}\label{sec:prelim}


We define $[n]=\{1,\ldots,n\}$. Let $x\in \R^m$, if $I\subseteq [m]$, then $x_{I}$ is the vector of length $|I|$ obtained by deleting elements outside the index set. $x(I) = \sum_{i\in I} x_i$. $\operatorname{Conv}(X)$ is the convex hull of $X$. 

\subsection{Prefix Sum and Piecewise Linear Convex Function Representations}
Given a sequence of elements $a_1,\ldots,a_n$ and some associative operation $\oplus$, the prefix sum is the sequence $b_1,\ldots,b_n$, such that $b_1=a_1$, and $b_i = b_{i-1} \oplus a_i$. The prefix sum data structure maintains the corresponding prefix sums under updates of the original sequence, allowing query of each prefix sum value and binary search (if monotonic) in $O(\log n)$ time \cite{blelloch1990}. 

Let $f:[0,\infty) \to \R$ be a piecewise linear convex function with $n$ breakpoints ($0$ is always a breakpoint). There are 3 different forms that capture almost all information of $f$, and one can transform between them using prefix sum or even easier operations.

\begin{enumerate}
    \item The slope-difference form $SD(f) = \{(x_1,\Delta_1),\ldots,\allowbreak (x_n,\Delta_n)\}$, where $x_1=0$, $\Delta_1$ is the left most slope of $f$, and $x_i$ is the $i$th breakpoint, and $\Delta_i$ for $i>1$ is the difference between the right slope and the left slope.
    \item The slope form $S(f) = \{(x_1,s_1),\ldots,(x_n,s_n)\}$. Again, $x_i$ are the breakpoints, and $s_i=\sum_{j=1}^i \Delta_j$ is the right slope at point $x_i$.
    \item The value form $V(f) = \{(x_1,f(x_1)-f(0)),\ldots, \allowbreak (x_n,f(x_n)-f(0))\}$.
\end{enumerate}

Note that the previous forms also require the value of $f(0)$ in order to uniquely recover the function, hence it has to be stored elsewhere. 
One can write the prefix sum data structure by hand, such that the original sequence is the slope-difference form, and any update in slope-difference form would propagate to slope and value form.

The slope-difference form is also easy for sums. $SD(f+g)$ is simply $SD(f)\cup SD(g)$ if $f$ and $g$ do not share breakpoints, otherwise, sum the slope-difference at the breakpoint. For simplicity of exposition, we assume the functions we sum do not share breakpoints. This also allows one to maintain $f=\sum_{i} f_i$ easily by taking the union.

By maintaining a function, we means that the following question can be answered quickly
\begin{enumerate}\label{enum:operations}
\item Evaluate: Given $x$, return $f(x)$.
\item Inverse: Given $y$, find smallest $x$ such that $f(x)=y$.
\item Output: Given $x$ and $y$, output the function $f$ restricted on $[x,y]$.
\end{enumerate}

If $f$ is the trade-off curve, then ``Evaluate'' can answer how much value can be obtained for a given budget, and ``Inverse'' can answer how much budget is required for a particular value.

\subsection{Matroids}

A matroid $M=(E,\mathcal{I})$ is a set system over ground set $E$, and $\mathcal{I}$ consists of subsets of $E$, such that the following properties hold. 

\begin{enumerate}
\item $\emptyset\in \mathcal{I}$.
\item $A\in \mathcal{I}$, then every subset of $A$ is in $\mathcal{I}$.
\item If $A,B\in \mathcal{I}$, and $|A|>|B|$, then there is $x\in A\setminus B$, such that $B\cup \{x\}\in \mathcal{I}$. 
\end{enumerate}

The sets in $\mathcal{I}$ are called \emph{independent sets}, and the maximal independent sets are called \emph{bases}. The \emph{rank} function $r$ associated with $M$ is defined as $r(S) =\max\{ |S'|  \mid S'\subseteq S ,S' \in \mathcal{I}\}$, the size of the largest independent set contained in $S$. The rank of the matroid is defined as $r(E)$. 

A matroid is a $p$-uniform matroid if there exists an integer $p$, such that a set is independent if and only if its size is at most $p$. 

\subsection{Problem and Properties}
Consider an (integral) incentive allocation problem with $n$ agents. The $i$th agent has a candidate set of incentives, $E_i$. Each incentive $e$ has a nonnegative cost $c_e$ and a nonnegative value $v_e$, respectively. To model constraints, let $\mathcal{F}_i$ be the feasible subsets of $E_i$, which can be encoded as a set of binary vectors. We assume $B\geq 0$ and $\emptyset\in\mathcal{F}_i$, and omit agents with no incentives. Let $m_i = |E_i|$ and $m=\sum_{i} m_i$. The problem is to choose a feasible set of incentives for each agent, such that the sum of value of all the chosen incentives is maximized while the total cost does not exceed budget $B$.

The (integral) incentive allocation problem can be formulated as the following integer program ($IP$):

\[
\begin{aligned}
    \max_x \quad v \cdot x \\
    s.t.   \quad c \cdot x &\leq B & \\
    x_{E_i}&\in \mathcal{F}_i &\forall i\in [n]\\
    x&\in \{0,1\}^m
\end{aligned}
\]

Define $\ips(B)$ to be the objective value of the above integer program. The exact trade-off curve is the function $\ips$ as $B$ ranges from $0$ to $\infty$. Finding $\ips(B)$ is NP-hard, therefore we consider the linear programming relaxation instead. This is shown below. 

\begin{equation}
    \label{eq:generallp}
\begin{aligned}
    \max_x \quad v \cdot x \\
    s.t.   \quad c \cdot x &\leq B & \\
    x_{E_i}&\in \operatorname{Conv}(\mathcal{F}_i) &\forall i\in [n]\\
\end{aligned}
\end{equation}

Define $\lps(B)$ to be the objective value of the linear programming relaxation of the integer program $IP$. We call $\lps$ the (fractional) \emph{trade-off curve}.

One can reduce the problem to multiple choice knapsack similar to the reduction by Pisinger \shortcite{DavidPisinger}, and show $\lps(B)-\ips(B) \leq \max_i\{ \sum_{e\in E_i} v_e\}$. 
That is, the maximum difference is at most the value a single agent can provide. If additionally, we know $\mathcal{F}_i$ forms a matroid for each $i$, then a stronger result exists: the difference is at most the value of a single incentive \cite{CAMERINI1984157}. Namely, $\lps(B)-\ips(B) \leq \|v\|_\infty$.

Because in large-scale problems such as coupon assignment, single agent's value is \emph{small} compared to the objective. Therefore, $\lps$ is a very close approximation of $\ips$. Hence, our work is to maintain the function $\lps$.

\section{Algorithm}\label{sec:alg}
The algorithm is conceptually simple. The computation gets broken into two independent parts, allowing for greater parallelization and customization. 

The idea is to compute a signature function for each agent. The signature functions can be computed in parallel, completely independently.
The sum of the signature functions is the function we will maintain, and we show how to use the sum to obtain the desired information on $\lps$.

\subsection{From Signature Functions to Trade-Off Curve}
\label{sec:sigf}

We start with the most general form of the problem \Cref{eq:generallp}. Let $P_i = \operatorname{Conv(\mathcal{F}_i)}$. Consider we have $n$ polyhedrons $P_1,\ldots,P_n$ together with disjoint index sets $E_1,\ldots,E_n$ with their union $[m]$.

Consider the following linear program, 
\[
    \begin{aligned}
    \max_x \quad v \cdot x \\
    s.t.   \quad c \cdot x &\leq B & \\
    x_{E_i} &\in P_i & \forall i\in[n]\\
    \end{aligned}
\]

For $\lambda\geq 0$, we define $f_i(\lambda) = \max\{(v_{E_i}-\lambda c_{E_i}) x | x\in P_i \}$, and we call it the \emph{signature function} of agent $i$. The signature function $f_i$ is piecewise linear and convex since it is the upper envelope of line arrangement $\{l_x(\lambda)=v_{E_i}\cdot x-\lambda c_{E_i}\cdot x | \forall x\in P_i\}$.

Let $f = \sum_{i} f_i$. 
The Lagrangian dual of the linear program is therefore

\begin{equation}
\label{eq:Lagrangiandual}
\begin{aligned}
\min_{\lambda\geq 0} \left( B\lambda+f(\lambda)\right).
\end{aligned}
\end{equation}

Note that each $f_i$ is a piecewise linear convex function, hence $\lambda B + f(\lambda)$ is also piecewise linear and convex.

Given the signature function $f_i$ for each agent, we have to maintain the function $f=\sum_{i} f_i$. Maintaining $f$ itself is an easy task since it is just the sum of piecewise linear functions, the number of breakpoints is at most the total number of breakpoints for $f_i$. If each $f_i$ is stored in slope-difference form, then $f$ can be computed through a simple merge of the lists.

\begin{theorem}\label{thm:lps}
    For $B\geq 0$, $\tau$ is a piecewise linear concave function and $\tau(B) = \min_{\lambda\geq 0} \{\lambda B+f(\lambda)\}$.
\end{theorem}

This shows once we have the signature functions, the trade-off curve is easy to compute through common techniques for manipulating piecewise linear functions. See the technical appendix for the full proof.

We have already established that $\lps$ is closely related to $f$. Next, we show how to maintain $\lps$ dynamically.

Assume $f$ has $k$ breakpoints and all three forms(value, slope and slope-difference) of $f$ are given. We answer the following questions.

\begin{enumerate}
    \item Evaluate: For a fixed $B$, how to find $\lps(B)$? 
    \item Inverse: Find a $B$ such that $\lps(B)=y$.
    \item Update: Maintain $\lps$ after a single agent's incentive changes. 
    \item Output: Output a contiguous piece of $\lps$.
\end{enumerate}

\paragraph{Evaluate.} $B\lambda + f(\lambda)$ is a piecewise linear convex function, the minimum is at the first breakpoint where the right slope is nonnegative. Or in other words, find the first $\lambda$ in $f$, where the right slope is at least $-B$. This can be processed easily by looking at the slope form of $f$ and do a binary search, and it would take $O(\log k)$ time.

\paragraph{Inverse.} Report no solution for $y\notin[\lps(0),f(0)]$; return $0$ for $y=\lps(0)$ and $-f'_+(0)$ for $y=f(0)>\lps(0)$. For an interior target, the idea is to find $\lps(B_1)\leq y< \lps(B_2)$, such that $B_1$ and $B_2$ are the negatives of the right and left slopes at a positive breakpoint of $f$, respectively. We can do binary search over the breakpoint of $f$ to find the corresponding $\lambda_1$ for $B_1$. Finally, solve the linear equation $B\lambda_1 + f(\lambda_1) = y$ to obtain $B$.

\paragraph{Update.}
Assume information for agent $i$ updates, then the only change is the signature function for that agent. Assume the new signature function after the update is $g_i$. The new $f$ is obtained by subtracting $f_i$ and adding $g_i$ successively. We could also save time by only updating the difference in $f_i$ and $g_i$. Hence, the amount of time spent on update is bounded by the number of breakpoint changes times a log factor, and the time finding the new signature function.

Note that the function is stored in slope-difference form; hence, any update in which only $t$ positions in the slope-difference form change takes $O(t\log k)$ time, which proves \Cref{thm:update}.

\paragraph{Output $\lps$.}
Evaluating and finding the inverse are sufficient for most purposes, but if we do want to output a contiguous piece of $\lps$ that consists of at most $t$ breakpoints, we can do so in $O(\log k+t)$ time. We know precisely for which $B$ the curve $\lps$ changes in slope: a one to one correspondence with the breakpoints of $f$. Hence, we can first find the desired place in $\lps$ and output the breakpoints one by one by walking through the slope table of $f$.

\begin{theorem}\label{thm:outeralgorithm}
Given the signature functions for each agent, it takes $O(k\log k)$ time to compute a representation for $\lps$, where $k$ bounds the number of breakpoints in the trade-off curve.
\end{theorem}

The running time in \Cref{thm:outeralgorithm} is obtained by merging signature functions of agents. The complexity of maintaining the trade-off curve depends on how many breakpoints there are in $f$, which in turn is linearly related to the number of breakpoints in each signature function. The next step is to bound the number of breakpoints in the signature functions, and the time to compute it.

Because $f$ decomposes as the sum of signature functions, we only have to focus on a single agent. So from this point on, we only consider the signature function for a single agent.

\subsection{General Signature Function}

In the most general case, we would define the signature function $f(\lambda)$ to be the optimum of $\max_x \{(v-\lambda c)\cdot x | x\in P\}$, where $x$ is an $m$ dimensional vector. This is the general parametric linear program. The number of breakpoints in $f$ can be exponentially large, namely $\Omega(2^{\sqrt{m}})$ \cite{Zadeh73b,Murty80,Carstensen83}.

However, if the constraints in the question are matroids, the number of breakpoints is reasonably small, and can be computed quickly. To start, we focus on the cardinality constrained case.

\subsection{Cardinality Constraint}

Consider an agent who has $m$ incentives $E$, and at most $p$ of them can be allocated to the agent. The signature function is $f(\lambda) = \max \{(v-\lambda c)\cdot x | \mathbf{1}\cdot x \leq p, 0\leq x\leq 1\}$. For ease of manipulation later, we actually want equality. That is, the agent gets exactly $p$ incentives. Indeed, we can add $p$ dummy incentives with $0$ value and $0$ cost. Pisinger observed the number of possible slopes is upper bounded by $m^2$, hence showing $f$ have at most $O(m^2)$ breakpoints \shortcite{DavidPisinger}. 

We use techniques from computational geometry to view this problem. Consider an arrangement of lines $\{\ell_e \mid e\in E\}$, where $\ell_e(\lambda) = v_e-\lambda c_e$ for $e\in E$. $f(\lambda)$ is the sum of $p$ top most lines when $x$ coordinate is $\lambda$. Therefore, in order to find $f$, it is sufficient to find the top $p$ lines in the arrangement for each $\lambda$. 
 
The simple brute force method is to first find all intersections of the lines, and sort them by $x$ coordinate. In between each two consecutive intersections, the top $p$ lines cannot change. So we calculate the top $p$ lines on all $x$ intervals formed by two consecutive intersections. 
Because there are $O(m^2)$ intersections, the number of breakpoints is also $O(m^2)$, which gives an alternative way to show Pisinger's bound. The bound is very loose, and next we show how the geometric view can improve the bound. The set of points that is the top $p$th point in an arrangement of lines is called the $p$-level \cite{ERDOS1973139,lovasz}. $1$-level is the upper envelope of the lines, which is the boundary of a convex space. However, for $p>1$, it is not necessarily convex, see \Cref{fig:2level}. The $p$-level can be constructed deterministically in $O(m\log m+mp^{1/3})$ time \cite{ChanChengZheng2024}. Observe that the slope of the signature function $f$ can only change at the breakpoint of the $p$-level. Indeed, even when there are many line intersections above the $p$-level, the top $p$ lines does not change, hence the sum would not change. The current best upper bound on the number of breakpoints of $p$-level is $O(mp^{1/3})$ \cite{Dey1998}. Together, it reflects $f$ has $O(mp^{1/3})$ breakpoints and can be computed in the same time as computing the $p$-level.

\begin{theorem}\label{thm:singlepsignature}
A signature function for $p\geq 1$ cardinality constrained incentive allocation can be computed in $O(m\log m+mp^{1/3})$ time. It has at most $O(mp^{1/3})$ breakpoints.
\end{theorem}

The true upper bound for number of breakpoints in $p$-level might be much smaller. The currently known lower bound is only $m 2^{\Omega(\sqrt{\log p})}$ \cite{Toth01}. Any improvement in the upper bound implies a better bound on the complexity of the trade-off curve.
\begin{figure}
    \centering
\includegraphics[width=.35\textwidth]{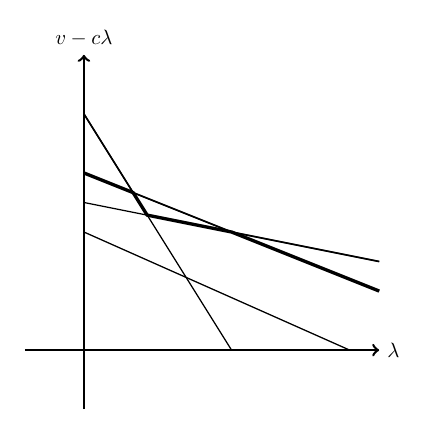}
\caption{The bold line forms a $2$-level in the line arrangement.}
\label{fig:2level}
\end{figure}
\subsection{Matroid Constraint} \label{matroid case}

The agent must be assigned an independent set of incentives in a matroid over ground set $E$. Let $|E|=m$, $r$ is the rank function, and $p=r(E)$ is the rank of the matroid.

Using standard knowledge from matroid theory \cite{Schrijver2003}, $f(\lambda)$ would be defined as the optimum of the following linear program, after adjoining $p$ free dummy items with 0 cost and 0 value and truncating to rank $p$; $E$ and $r$ below refer to this augmented matroid.
\begin{align*}
\max_x \quad & (v-\lambda c)\cdot x\\
s.t. \quad 
x(S)&\leq r(S)  \quad\forall S\subseteq E\\
x(E)&=p\\
x&\geq 0
\end{align*}

For a fixed $\lambda$, this LP is finding the optimum weight base in a matroid, where the weight is $w(e)=v_e-\lambda c_e$. 
Breakpoints on $f(\lambda)$ indicate that the matroid's optimum base changes due to the linear change in weights.
The number of breakpoints on $f$ is bounded by $O(mp^{1/3})$ \cite{Dey1998}. 
Unfortunately, unlike the cardinality case, there are matroids forcing $\Omega(mp^{1/3})$ breakpoints on the signature function \cite{Eppstein98}.


To compute $f$, we need to find all breakpoints on $f$. However, for matroid constraints there is no existing efficient algorithm finding breakpoints on the signature functions for general matroids as the $k$-level algorithm for cardinality constraints. 

Similar problem under graphic matroid has been studied in \cite{agarwal_parametric_1998}. The authors use parametric search and sparsification techniques to find breakpoints efficiently. 
The techniques are limited to graphs and cannot be applied to general matroids.
We achieve a running time of $O(Tmp^{1/3})$ for general matroids by using the Eisner-Severance method, where $T$ is the time complexity of finding optimum weight base in a matroid.

Eisner-Severance method is a simple algorithm for finding breakpoints on convex piecewise linear functions \cite{eisner_mathematical_1976}.
Given any piecewise linear convex function $f:[0,\infty)\to\R$ with $k$ breakpoints and an oracle which computes $f(\lambda)$ and arbitrary tangent line of $f$ to the right of $\lambda$.
ES method finds all breakpoints on $f$ with $O(k)$ oracle calls. The method is as follows.
Let $\{l_1,\ldots,l_q\}$ be the complete ordered list of affine pieces of $f$; we maintain a discovered subsequence $L$.
Initially, the sequence $L=\{l_1,l_q\}$ contains the leftmost and rightmost segments, found, for our signatures, at $0$ and by optimizing over zero-cost independent sets, respectively. If they coincide, the algorithm stops.
Denote by $\Lambda$ the list of intersections of adjacent lines in $L$.
ES method works by repeatedly adding line segments to $L$. In each iteration we check one intersection $\lambda_i\in\Lambda$ and evaluate $f(\lambda_i)$. Suppose $\lambda_i$ is the intersection of two adjacent stored line segments $l_a$ and $l_b$.
For consecutive pieces $l_t,l_{t+1}$ in the complete list, $\lambda$ is their common breakpoint if and only if $f(\lambda)=l_t(\lambda)=l_{t+1}(\lambda)$. For the partial list $L$, the test $f(\lambda)=l_a(\lambda)=l_b(\lambda)$ certifies that the stored pieces are consecutive in the complete list. If this equality holds, we remove $\lambda$ from list $\Lambda$; Otherwise, there exists a new line segment $l_p$ that attains $f(\lambda)$ and is not yet in $L$ and can be found using the oracle. We insert $l_p$ to $L$ and replace the tested intersection in $\Lambda$ by its intersections with adjacent lines. The algorithm terminates when $\Lambda=\emptyset$. The correctness of the algorithm is ensured by the correctness of the ES method.

Each intersection added to $\Lambda$ gives us a breakpoint or a new line segment. Thus the total number of evaluations of $f$ is $O(k)$, where $k=O(mp^{1/3})$ is the number of breakpoints.

For finding $l_p$ and evaluate $f(\lambda)$, we need to find the optimal weight base which takes $O(T)$. Thus the total time complexity of computing signature function for one agent with $m$ incentives is $O(Tm p^{1/3})$.



Our algorithm also leads to a simple proof that the optimal solution has at most two fractional variables. This fact can also be deduced from the proof of the integrality gap \cite{CAMERINI1984157}.
\begin{theorem}\label{thm:2frac}
There exists an optimal solution to \Cref{eq:generallp} under matroid constraints with at most $2$ fractional variables.
\end{theorem}

See the technical appendix for the proof.

\subsection{Wrapping Up}

Combining \Cref{thm:singlepsignature} and \Cref{thm:outeralgorithm}, we obtain the desired theorems.
\cardinality*
\begin{proof}
Assume the $i$th agent has $m_i$ choices of incentives, with $\sum_i m_i=m$. By \Cref{thm:singlepsignature}, computing all signature functions takes $O(\sum_i(m_i\log m_i+m_i p^{1/3}))=O(m\log m+mp^{1/3})$ time. Their total number of signature records is $O(mp^{1/3})=O(k)$, so \Cref{thm:outeralgorithm} adds $O(k\log k)=O(k\log m)$ time, taking $p\leq m$. The total time is therefore $O((m+k)\log m)=O(mp^{1/3}\log m)$.
\end{proof}
By the above theorem, when $p=1$, namely the multiple choice constrained case, $k=O(m)$, and we obtain the desired $O(m\log m)$ running time.

For matroid constraints, we get a more modest result.
\matroid*
\begin{proof}
    Assume the signature function has $k_i$ breakpoints for agent $i$ and $m$ is the total number of incentives. The running time of computing one signature function is $O(Tk_i)$. 
    Computing all signature functions takes $O(Tk)$ since $\sum_i Tk_i\leq Tk$. 
    By \Cref{thm:outeralgorithm}, constructing the data structure for $\lps$ takes $O(k\log k) = O(k\log m)$ time. So together we have the running time $O( Tk+ k\log m) $. 
    \end{proof}

For practical purpose, once we have the signature functions for single agents, the trade-off curve can be easily computed with OLAP databases. See the Technical Appendix for details.

Next we discuss the work per update of a single agent. Each single agent update can only change breakpoints of the associated signature function. If at most $s$ incentives are related to the agent before and after the update, at most $O(s^{4/3})$ breakpoint changes can happen, The maintenance time would be $O(s^{4/3}\log k)$ after obtaining the edits; including signature recomputation gives $O(Ts^{4/3}+s^{4/3}\log k)$ for a general matroid. As we assumed in the scenario, $s$ is small because no agent is related to too many incentives, hence this would be a fast operation in modern systems.

\section{Submodular Objective}

In this section we discuss a more general case where the objective function is submodular instead of linear.
Submodular objective function reflects the diminishing marginal gain phenomenon thus is closer to reality. In practice, agents usually receive incentives for free. We further assume that the submodular objective function $g:2^E\to \R$ is monotone, non-negative and satisfies $g(\emptyset)=0$. Thus, we are particularly interested in polymatroid objective functions.

The submodular incentive allocation problem can be formulated as follows:
\begin{align*}
    \max_x \; g(x) \\
    s.t.   \quad c \cdot x &\leq B \\
    x_{E_i}&\in \mathcal{F}_i \quad\forall i\in [n]\\
    x&\in \{0,1\}^m
\end{align*}

where $g:\{0,1\}^m \to \R$ is a polymatroid set function.



We define the signature function for agent $i$ to be $f_i(\lambda) = \max\left\{g(x)-\lambda c\cdot x | x\in \mathcal{F}_i \cap \{0,1 \}^{|E_i|}\right\}$. The Lagrangian dual can be written as $\min_{\lambda\geq 0} B\lambda + \sum_i f_i(\lambda)$. 

Note that the properties of signature functions in \Cref{sec:sigf} are independent of the objective function and constraints. Therefore, $f_i(\lambda)$ is piecewise linear and convex, even for submodular objectives. However, our algorithm does not extend to the submodular case.

Our method requires an efficient algorithm for evaluating $f_i(\lambda)$.
For the submodular case, we need to solve a constrained submodular maximization problem to compute $f_i(\lambda)$. 
It is known that this problem is NP-hard \cite{calinescu_maximizing_2011}, so we consider solving it approximately. $g(x)-\lambda c\cdot x$ is still submodular for $x$ but is not monotone. 
The best approximation rate is $g(x)-\lambda \cdot x \geq (1-1/e)g(x_{OPT})-\lambda \cdot x_{OPT}$ in \cite{sviridenko_optimal_2017}.
However, the running time is impractical for implementations and currently no nontrivial upper bound is known for the number of breakpoints on $f_i$.

\section{Computational Results}
Our paper is mostly theoretical, but we did an implementation to see how does theory fair in practice for the \emph{cardinality constraint} case. Do we need to use advanced computational geometry tools to obtain good result in practice?

For the cardinality case we implemented two algorithms. The first uses the kinetic heap method described in Section~2.3 of \cite{Chan1999RemarksOK} to construct the $p$-level in $O(m\log m+k\log^{1+\varepsilon}m)$ time and $O(m)$ working space, for any fixed $\varepsilon>0$, where $k$ counts processed level vertices. The other is a simple scan line algorithm. The algorithm maintains the $p$-level by looking at all intersections with the current $p$th line, which gives an $O(km)$ running time. These benchmarks do not include the newer algorithm of \cite{ChanChengZheng2024}.

All tests were run on MacOS operating system with an M2Max cpu. 
\Cref{tab:klevel} shows the average running time of 10 random instances each case, the numbers are drawn from a uniform sample.

\begin{table}[!ht]
    \centering
    \begin{tabular}{ccccc}
        \toprule
         \multirow{2}*{$m$} & \multicolumn{2}{c}{$p=20$} & \multicolumn{2}{c}{$p=40$} \\
         \cmidrule(lr){2-3}  \cmidrule(lr){4-5} 
         & scan & opt & scan & opt \\
         \midrule
        $1\times 10^3$            & 0.000 & 0.000 & 0.000 & 0.001 \\
        $5\times 10^3$    & 0.003 & 0.005 & 0.006 & 0.005 \\
        $1\times 10^4$            & 0.008 & 0.010 & 0.014 & 0.012 \\
        $5\times 10^4$    & 0.043 & 0.089 & 0.080 & 0.087 \\
        $1\times 10^5$            & 0.094 & 0.216 & 0.173 & 0.223 \\
        $5\times 10^5$    & 0.528 & 2.911 & 0.937 & 2.952 \\
        $1\times 10^6$            & 1.147 & 7.291 & 1.989 & 7.140 \\
        $1\times 10^7$            & 12.994 & 100.512 & 23.863 & 101.675 \\
        
        \bottomrule
    \end{tabular}
    \begin{tabular}{ccccc}
         \multirow{2}*{$m$} & \multicolumn{2}{c}{$p=2000$} & \multicolumn{2}{c}{$p=m/5$}\\
         \cmidrule(lr){2-3}  \cmidrule(lr){4-5} 
         & scan & opt & scan & opt \\
         \midrule
        $1\times 10^3$      & - & - & 0.003& 0.002 \\
        $5\times 10^3$      & 0.137 & 0.027 & 0.091& 0.02\\
        $1\times 10^4$      & 0.384 & 0.048 & 0.384 & 0.048\\
        $5\times 10^4$      & 2.634 & 0.187 & 9.531& 0.326\\
        $1\times 10^5$      & 5.795 & 0.397 & 38.275& 1.222\\
        $5\times 10^5$      & 33.760 & 3.398 & TLE & 10.500 \\
        $1\times 10^6$      & 72.485 & 7.604 & TLE & 23.203\\
        $1\times 10^7$      & TLE & 101.775 & TLE & 133.974\\
        
        \bottomrule
    \end{tabular}
    \caption{The time (in seconds) to compute the breakpoints on the signature function under cardinality constraint using the kinetic heap method (opt) and the scan line algorithm (scan).}
    \label{tab:klevel}
\end{table}

The scan line algorithm is surprisingly good for small $p$. It is because in those cases $k$ is actually very small, much smaller than $m$. There is an intuitive argument. If $p=1$, then $k$ is bounded by the number of points on the convex hull of a uniform random sample of $m$ points in a rectangle. The expected value is $O(\log m)$ \cite{randompoint}. Note as $p$ becomes larger, for a random set of points $k$ also becomes larger, and therefore the $O(mk)$ algorithm suffers. When $k=O(m)$, the kinetic heap method has a running time of $O(m\log^{1+\varepsilon}m)$ for any fixed $\varepsilon>0$.

For the matroid case we tested our algorithm on laminar matroids. 
The laminar matroid is defined on a laminar family. Given a set $E$, a family $\mathcal{A}$ of subsets of $E$ is \emph{laminar} if for every two sets $A,B\in \mathcal{A}$ with $A\cap B\not= \emptyset$, either $A\subseteq B$ or $B\subseteq A$. Define the capacity function $c: \mathcal{A} \rightarrow \mathbb{R}$. The independent set $\mathcal{I}$ of a laminar matroid $\mathcal{L}$ is the set of subsets $I$ of $E$ such that $|I\cap A|\leq c(A)$ for all $A\in \mathcal{A}$ \cite{fife_laminar_2017}. 
We implemented the Eisner-Severance method on laminar matroids for demonstration purposes. \Cref{tab:matroid} shows the average running time for computing the signature function under laminar matroid constraints.

\begin{table}[!t]
  \centering
  \resizebox{\columnwidth}{!}{
  \begin{tabular}{cc|cc|cc}
      \toprule
       $m$ & $t$ & $m$ & $t$ & $m$ & $t$ \\
       \midrule
       $1\times 10^3$         & 0.0161 & $1.1\times 10^4$ & 1.5270 & $2.5 \times 10^4$ & 6.8601 \\
       $2\times 10^3$ & 0.0575 & $1.2\times 10^4$ & 1.8602 & $3 \times 10^4$   & 7.8284 \\
       $3\times 10^3$ & 0.1375 & $1.3\times 10^4$ & 1.8959 & $3.5 \times 10^4$ & 12.1495 \\
       $4\times 10^3$ & 0.2093 & $1.4\times 10^4$ & 2.3682 & $4 \times 10^4$   & 15.6755 \\
       $5\times 10^3$ & 0.3547 & $1.5\times 10^4$ & 2.4609 & $4.5 \times 10^4$ & 18.9251 \\
       $6\times 10^3$ & 0.5193 & $1.6\times 10^4$ & 2.7309 & $5 \times 10^4$   & 25.0841 \\
       $7\times 10^3$ & 0.6469 & $1.7\times 10^4$ & 3.1121 & $5.5 \times 10^4$ & 24.6682 \\
       $8\times 10^3$ & 0.7878 & $1.8\times 10^4$ & 3.7226 & $6 \times 10^4$   & 26.5710 \\
       $9\times 10^3$ & 1.0582 & $1.9\times 10^4$ & 4.3983 & $6.5 \times 10^4$ & 34.9471 \\
       $1\times 10^4$        & 1.2360 & $2 \times 10^4$  & 4.2026 & $7 \times 10^4$   & 44.8108 \\
      \bottomrule
  \end{tabular}
  }
  \caption{The time (in seconds) to compute the signature function under matroid constraint.}
  \label{tab:matroid}
\end{table}

\section*{Acknowledgments}
This work was supported by the National Natural Science Foundation of China under grant 62372093, and by Science and Technology Department of Sichuan Province under grant M112024ZYD0170.

\bibliographystyle{named}
\bibliography{ijcai25}
\end{document}